%% file: klaska_195.tex
\documentclass[mathfont=newtx,accepted]{uai2021}

\usepackage{natbib}
\usepackage{booktabs}
\usepackage{multirow}
\usepackage[linesnumbered,ruled]{algorithm2e}

\usepackage[a-2u]{pdfx}

\usepackage{tikz}
\tikzstyle{stoch}=[circle,thick,draw,minimum size=1.5em,inner sep=0em]
\tikzstyle{max}=[rectangle,thick,draw,minimum size=1.5em,inner sep=0em]
\tikzstyle{tran}=[thick,draw,->,>=stealth,rounded corners]
\tikzstyle{loop left}=[tran, to path={.. controls +(150:.8) 
	and +(210:.8) .. (\tikztotarget) \tikztonodes}]
\tikzstyle{loop right}=[tran, to path={.. controls +(30:.8) 
	and +(330:.8) .. (\tikztotarget) \tikztonodes}]
\tikzstyle{loop above}=[tran, to path={.. controls +(60:.5) 
	and +(120:.5) .. (\tikztotarget) \tikztonodes}]
\tikzstyle{loop below}=[tran, to path={.. controls +(240:.8) 
	and +(300:.8) .. (\tikztotarget) \tikztonodes}]

\input{macros}

\title{Regstar: Efficient Strategy Synthesis for Adversarial Patrolling Games}

\author[1]{David Kla\v{s}ka}
\author[1]{Anton\'{\i}n Ku\v{c}era}
\author[1]{V\'{\i}t Musil}
\author[1]{Vojt\v{e}ch {\v{R}}eh\'{a}k}

\affil[1]{%
Faculty of Informatics,
Masaryk University,
Brno, Czech Republic}

\begin{document}

\maketitle

\begin{abstract}
We design a new efficient strategy synthesis method applicable to adversarial patrolling problems on graphs with arbitrary-length edges and possibly imperfect intrusion detection.
The core ingredient is an efficient algorithm for computing the value and the gradient of a function assigning to every strategy its ``protection'' achieved.
This allows for designing an efficient strategy improvement algorithm by differentiable programming and optimization techniques.
Our method is the first one applicable to real-world patrolling graphs of reasonable sizes.
It outperforms the state-of-the-art strategy synthesis algorithm by a margin.
\end{abstract}

\section{Introduction}
\label{sec-intro}

\emph{Patrolling games} are a special type of security games where a mobile
Defender moves among vulnerable targets and aims to detect possible
ongoing intrusions initiated by an Attacker. The targets are modelled as
vertices in a directed graph where the edges correspond to admissible
Defender's moves. At any moment, the Attacker may choose some target
$\tau$ and initiate an intrusion (attack) at~$\tau$. 
Completing this intrusion takes $d(\tau)$ time units,
and if the Defender does not visit $\tau$ in time,
he is penalized by utility loss determined by the cost of~$\tau$.

In \emph{adversarial} patrolling games
\citep{VAT:adversarial-patrolling,AKK:multi-robot-perimeter-adversarial,Agmon2009,BGA:large-patrol-AI,Basilico2009,MunozdeCote2013,SLESSLIN2019},
the Attacker knows the Defender's strategy and can even observe the Defender's
moves\footnote{The Defender may choose the next move randomly according to a
distribution specified by its moving strategy. Although the Attacker knows the
Defender's strategy (\ie, the distribution), it \emph{cannot} predict the way
of resolving the randomized choice.}. These assumptions are particularly
appropriate in situations where the actual Attacker's abilities are
\emph{unknown} and the Defender is obliged to guarantee a certain level of
protection even in the worst case. This naturally leads to using
\emph{Stackelberg equilibrium}
\citep{DBLP:conf/ijcai/SinhaFAKT18,yin2010stackelberg} as the underlying
solution concept, where the Defender/Attacker play the roles of the
leader/follower, \ie, the Defender commits to a moving strategy $\gamma$, and
the Attacker follows by selecting an appropriate counter-strategy~$\pi$. The \emph{value} of $\gamma$, denoted by~$\Val(\gamma)$, is the expected Defender's utility guaranteed by $\gamma$ against an arbitrary Attacker's strategy. Intuitively, $\Val(\gamma)$ corresponds to the ``level of protection'' achieved by~$\gamma$.

\begin{figure}[b]
	\centering
	\includegraphics[width=\linewidth]{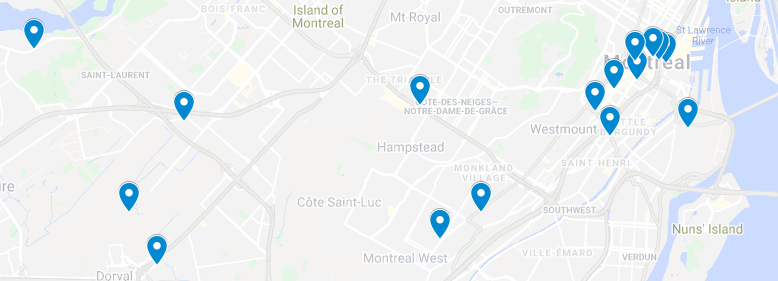}
	\caption{We synthesize an efficient strategy for patrolling the net of Montreal's ATMS.}
\label{F:atms}
\end{figure}

The basic algorithmic problem in patrolling games is to compute a Defender's strategy $\gamma$ such that $\Val(\gamma)$ is as large as possible. Since general history-dependent strategies are not algorithmically workable (see \sect{sec-Defender-strategy}), recent works \citep{KL:patrol-regular,DBLP:conf/atal/KlaskaKLR18} concentrate on computing \emph{regular} strategies where the Defender's decisions depend on finite information about the history of previously visited vertices. As \citet{KL:patrol-regular} observed, regular strategies provide better protection than \emph{memoryless} strategies where the Defender's decision depends only on the currently visited vertex. However, the mentioned algorithms apply only to patrolling graphs where all edges have the \emph{same} length (traversal time). A longer distance between vertices can be modelled only by adding a sequence of auxiliary vertices and edges, quickly pushing the resulting graph's size beyond the edge of feasibility (see \sect{sec-experiments}). Even the currently best algorithm of \citet{DBLP:conf/atal/KlaskaKLR18} fails to solve small real-world patrolling graphs such as a network of selected ATMs in Montreal (\fig{F:atms}).

\paragraph{Our contribution}

\begin{enumerate}
\item We prove that \textbf{regular Defender's strategies} are not only better than memoryless strategies, but they are \textbf{arbitrarily close} to the power of general strategies.
Therefore, the scope of strategy synthesis can be safely restricted to regular strategies.
This resolves the open question of previous works.

\item We design an efficient \Regstar\footnote{REGular STrategy ARchitect.} algorithm for computing regular Defender's strategies in general patrolling graphs with edges of \textbf{arbitrary length}. 
\Regstar applies to scenarios with \textbf{imperfect intrusion detection}, where the probability of discovering an intrusion at a target~$\tau$ by the Defender visiting $\tau$ is not necessarily equal to one.

\item We validate \Regstar experimentally.
We compare \Regstar against the best existing algorithm when applicable, and
we perform a set of \textbf{real-life experiments} demonstrating its strengths and limits.
\end{enumerate}

We depart from the fact that the function $\Val$ assigning the protection value to a given regular Defender's strategy is \emph{differentiable}.
The very heart of our algorithm is a novel, efficient procedure for computing the value and the gradient of $\Val$.
Since the size of the closed-form expression representing $\Val$ is exponential, the task is highly non-trivial.
We apply differentiable programming techniques and design an efficient strategy improvement algorithm for regular strategies based on gradient ascent.
The \Regstar algorithm randomly generates many initial regular strategies, improves them, and returns the best strategy.

The efficiency of \Regstar is evaluated experimentally in \sect{sec-experiments}.
In the first series of experiments, we compare the efficiency of \Regstar against the algorithm of~\citet{DBLP:conf/atal/KlaskaKLR18}.
In the second series, we demonstrate the applicability of \Regstar to a real-world patrolling graph with 18 targets corresponding to selected ATMs in Montreal.
In the last series, we document the power of regular strategies on patrolling graphs modelling buildings with corridors and offices.
Here, the information about the history of visited vertices is crucial for achieving reasonable protection.   

Experiments prove that \Regstar outperforms the method of~\citet{DBLP:conf/atal/KlaskaKLR18} and can solve instances far beyond the reach of this algorithm.
Our approach adopts the \emph{infinite horizon} patrolling game model, and therefore it does not suffer from the scalability issues caused by increasing the time bound in finite-horizon security games (see \sect{sec-related} for more comments).
For practical applications, solving patrolling graphs with about 20--30 targets seems sufficient, as the protection achievable by a single Defender becomes low for a higher number of targets, and the patrolling task needs to be solved by multiple Defenders\footnote{Patrolling by multiple Defenders is studied independently, see, \eg, \citep{Beynier:multiagent-patrol,GEW:multiple-defenders}}.

\section{Related Work}
\label{sec-related}

Patrolling games are a special type of \emph{security games} where game-theoretic concepts are used to determine the optimal use of limited security resources \citep{Tambe:book}.
Security games with static allocation have been studied in, \eg, \citep{JKKOT:massive-security-games,JKKOT:optimal-resource-massive-security-games,PJMOPTWPK:Deployed-ARMOR,TRKOT:IRIS,DBLP:conf/aaai/XuWVT18,DBLP:conf/ijcai/XuJSRDT15,DBLP:conf/aaai/GanAVG17}.
For patrolling games, where the Defender is mobile, most of the existing works assume the Defender is following a \emph{positional} strategy that depends solely on the current position of the Defender \citep{BGA:large-patrol-AI}.
Since positional strategies are weaker than general history-dependent strategies, there were also attempts to utilize the history of the Defender's moves.
This includes the technique of duplicating each node of the graph to distinguish internal states of the Defender (for example, \citet{AKK:multi-robot-perimeter-adversarial} consider a direction of the patrolling robot as a specific state; this is further generalized in \citet{bosansky2012-aaaiss}).
Another concept is higher-order strategies \citep{Basilico2009}, where the Defender takes into account a bounded sequence of previously visited states.
In our work, we use regular strategies \citep{KL:patrol-regular,DBLP:conf/atal/KlaskaKLR18}, where the information about the history of Defender's moves is abstracted into finitely many memory elements. 

The existing strategy synthesis algorithms for patrolling games are based either on (1) mathematical programming with non-linear constraints, or (2) restricting the graph topology to some manageable subclass, or (3) strategy improvement, or (4) reinforcement learning.
The first approach (see, \eg, \citep{BGA:large-patrol-AI,Basilico2009,bosansky2011aamas}) suffers from scalability issues. In \citet{bosansky2011aamas}, the authors consider mobile targets, which forces the strategy to be time-dependent. \citet{BGA:large-patrol-AI} consider higher-order strategies in theory, but they perform experiments with positional strategies only due to computational infeasibility (one can easily construct examples where positional strategies are weaker than regular strategies and the protection gap is up to 100\%. The experiments of \citet{VAT:adversarial-patrolling} are also limited to positional strategies. \citet{SLESSLIN2019} make full use of the history, yet they study only perimeters (\ie, cycles), so their approach does not apply to graphs with arbitrary topologies.

The second approach applies only to selected topologies, such as lines, circles \citep{AKK:multi-robot-perimeter-adversarial,AgmonKK08-2}, or fully connected graphs with unit distance among all vertices \citep{BKR:patrol-Internet}.
The third approach has been applied only to patrolling graphs with edges of unit length.
The algorithm of \citet{KL:patrol-regular} requires a certain level of human assistance because the underlying finite-state automaton gathering the information about the Defender's history must be handcrafted.
\citet{DBLP:conf/atal/KlaskaKLR18} overcome this limitation by designing an automatic strategy synthesis algorithm, and it is the most efficient strategy synthesis procedure existing before our work. There is a high-level similarity to our algorithm because both use a variant of gradient ascent and construct regular strategies. However, the internals of the two algorithms is different. The core of our method is a novel procedure for computing the value and the gradient of the value function. In particular, our algorithm avoids the blowup in the number of states when modeling real-world patrolling scenarios with variable length edges. This allows for processing instances far beyond the reach of the algorithm of \citet{KKLR:patrol-gradient}, as documented experimentally in \sect{sec-Montreal}.

The fourth approach has been successful mainly for games with the finite horizon \citep{WSYWSJF:Patrolling-learning,DBLP:conf/aaai/KarwowskiMZGA19} and suffers from the exponential blowup in the number of finite paths with the increasing time bound.

\section{Patrolling Games}
\label{sec-patrolling}

We recall the standard model of adversarial patrolling games and fix the notation used.

\paragraph{Patrolling graphs}
\label{sec-graphs}

A \emph{patrolling graph} is a tuple $G = (V,T,E,\tm,d,\alpha,\beta)$, where
\begin{itemize}
	\item $V$ is a non-empty set of \emph{vertices};
	\item $T \subseteq V$ is a non-empty set of \emph{targets}; 
	\item $E \subseteq V \times V$ is a set of \emph{edges};
	\item $\tm\colon E\to\Nset$ denotes the time to travel an edge;
	\item $d\colon T\to\Nset$ specifies the time to complete an attack;
	\item $\alpha\colon T\to\Rset_+$ defines the costs of targets;
	\item $\beta\colon T\to(0,1]$ is the probability of a successful intrusion detection.
\end{itemize}
For short, we write $u\to v$ instead of $(u,v)\in E$, and denote $\alpha_{\max}=\max_{\tau\in T} \alpha(\tau)$ and $\dm=\max_{t\in T} d(t)$.  
In the sequel, let $G$ be a fixed patrolling graph.

\subsection{Defender's strategy}
\label{sec-Defender-strategy}

A Defender's strategy is a recipe for selecting the next vertex.  In general,
the Defender may choose the next vertex randomly depending on some information
about the history of previously visited vertices. 

Let $\calH$ be the set of all finite paths in $G$, including the empty path~$\lambda$.
A \emph{Defender's strategy} for $G$ is a function $\gamma\colon\calH
\rightarrow \Dist(V)$ where $\Dist(V)$ is the set of all probability
distributions on~$V$ such that whenever $\gamma(h)(v) > 0$, then either $h = \lambda$ or $u \rightarrow v$ where $u$ is the last vertex of~$h$. Note that $\gamma(\lambda)$ corresponds to the initial distribution on $V$.

Unrestricted Defender's strategies may depend on the whole history of previously visited vertices when
selecting the next vertex, and they may not be finitely representable.

\subsection{Attacker's strategy}

We consider the same patrolling game as in \citet{KKR:patrol-drones}, where the time is spent by moving along the edges.
We also use the same notion of Attacker's strategy, assuming that, in the \emph{worst case}, the Attacker can determine the next edge taken by the Defender \emph{immediately} after the Defender leaves the currently visited vertex.
This means that the Attacker's decision is based not only on the history of visited vertices but also on edge taken next.

The Attacker cannot gain anything by delaying his attack until the Defender arrives
at the next vertex.
Therefore we can assume an attack is initiated at the moment
when the Defender leaves the currently visited vertex.
Furthermore, the Attacker can attack \emph{at most once} during a play\footnote{Even if the Attacker can perform another attack after completing the previous one, the best choice the Defender is to follow an optimal strategy constructed for the single attack scenario.
This is no longer true when the Defender has to spend some time responding to the discovered attack \citep{SLESSLIN2019} or when multiple Attackers can perform several attacks concurrently.}.

An \emph{observation} is a finite sequence $o =
v_1,\ldots, v_n$, $v_n {\rightarrow} v_{n+1}$, where $v_1,\ldots, v_n\in\calH$.
The set of all observations is denoted by $\Obs$. 
An \emph{Attacker's strategy} for a patrolling graph $G$ is a function
$\pi\colon \Obs \rightarrow \{\wait,\attack_\tau:\tau\in T\}$. We require that
if $\pi(v_1,\ldots, v_n,v_n {\rightarrow} u) = \attack_\tau$ for some $\tau \in T$,
then $\pi(v_1,\ldots, v_i,v_i {\rightarrow} v_{i+1}) = \wait$ for all $1\leq
i<n$ ensuring that the Attacker can attack at most once.

\subsection{Evaluating Defender's strategy}

\newcommand\Eval{\operatorname{Eval}}

Let $o = v_1,\ldots, v_n$, $v_n {\rightarrow} v_{n+1}$, $n\ge 1$, be an observation, $\tau\in T$ a target and consider an attack at $\tau$ after observing~$o$.

By $\Path(o,\tau)$ we denote the set of all finite paths $u$ from $v_{n+1}$ to $v_{n+1+k}=\tau$ of the length $k\ge 0$ such that the total time needed to traverse from $v_n$ to $\tau$ along $u$ does not exceed~$d(\tau)$.

For $u\in\Path(o,\tau)$, let $\Eval(\tau\mid u)$ be the value defended at $\tau$ when the Defender discovered the attack in the last vertex of $u$. 
The probability of detecting the attack at $\tau$ after traversing $u$ equals $(1-\beta(\tau))^{\#_\tau(u)-1} \cdot \beta(\tau)$, where $\#_\tau(u)$ stands for the number of visits to $\tau$ along $u$.
Indeed, since the intrusion detection is not perfect, the Defender failed to discover the attack at the first $\#_\tau(u)-1$ trials with the probability $1-\beta(\tau)$ and succeeded at the last one with the probability~$\beta(\tau)$.
Specially, when $\beta=1$ and $\#_\tau(u)=1$, the factor becomes $0^0$, interpreted as~$1$. 
Therefore, as $\alpha(\tau)$ denotes the cost of $\tau$,
\begin{equation}
	\Eval(\tau\mid u)
		= \alpha(\tau) \cdot (1-\beta(\tau))^{\#_\tau(u)-1} \cdot \beta(\tau).
\end{equation}

The \emph{protection} achieved by $\gamma$ against an attack at~$\tau$ initiated after observing~$o$ is defined as 
\begin{equation}
  \Protect^{\gamma}(\tau\mid o)
		= \!\!\sum_{u \in \Path(o,\tau)} \Prob^\gamma(u\mid o) \cdot \Eval(\tau\mid u),
\end{equation}
where $\Prob^\gamma(u \mid o)$ is the probability of performing~$u$ after observing~$o$, \ie,
\begin{equation}
	\Prob^\gamma(u\mid o)
		= \prod_{i=1}^{k} \gamma(v_1,\ldots,v_{n+i})(v_{n+i+1}).
\end{equation}
Similarly, we use $\Prob^\gamma(o)$ to denote the probability that observation~$o$ occurs, \ie,
\begin{equation}
\Prob^\gamma(o)
= \prod_{i=0}^{n} \gamma(v_1,\ldots,v_{i})(v_{i+1}).
\end{equation}
Now let $\pi$ be an Attacker's strategy.
For every $\tau \in T$, let $\Attack(\pi,\tau)$ be the set of all $o\in\Omega$ such that $\pi(o)=\attack_\tau$.
The \emph{expected Attacker's utility} for $\gamma$ and $\pi$ is defined as
\begin{equation}
	\EU_A(\gamma,\pi)
	= \sum_{\tau \in T} \sum_{o\in\Attack(\pi,\tau)}
				\!\!\Prob^{\gamma}(o) \cdot \bigl[\alpha(\tau) - \Protect^{\gamma}(\tau\mid o)\bigr].
\end{equation}
Note that $\EU_A(\gamma,\pi)$ corresponds to the expected amount ``stolen'' by the Attacker.
Consistently with \citet{DBLP:conf/atal/KlaskaKLR18}, the \emph{expected Defender's utility} is defined as 
\begin{equation} \label{eq-EUD}
   \EU_D(\gamma,\pi)
	 	= \alpha_{\max} - \EU_A(\gamma,\pi)\,.
\end{equation}


The Defender/Attacker aims to maximize/minimize the expected protection $\EU_D(\gamma,\pi)$, respectively.
The \emph{value} of a given Defender's strategy $\gamma$  is the expected protection guaranteed by $\gamma$ against an \emph{arbitrary} Attacker's strategy, \ie,
\begin{equation} \label{E:val-def}
	\Val_G(\gamma)  =  \inf_{\pi} \ \EU_D(\gamma,\pi)\,.
\end{equation}
Maximal protection achievable in $G$ is then
\begin{equation}
  \Val_G = \sup_\gamma \Val_G(\gamma)\,.
\end{equation}

\section{The Method}

Our approach consists of three stages.
First, since the value function from \eqref{E:val-def} is not a priori
tractable, we analyze the value of \emph{regular strategies} for $G$.
This yields a closed-form differentiable value function.
Secondly, we design an \emph{efficient algorithm} that computes the value function and its gradient. This is the heart of our contribution.
Finally, we combine these elements into the \Regstar algorithm, which \emph{generates the best regular strategies} via gradient ascent.

\subsection{Regular Defender's strategies}
\label{sec-reg-strategies}

An algorithmically workable subclass of Defender's strategies are
\emph{regular} strategies used by~\citet{KL:patrol-regular,DBLP:conf/atal/KlaskaKLR18}.
In this concept, the information about the history of visited
vertices is abstracted into a \emph{finite} set of memory elements assigned to each vertex.

Formally, we turn vertices $v\in V$ of $G$ into \emph{eligible pairs}
$\dhat{V}=\{(v,m): v\in V,\,1\le m\le \mem(v)\}$
in which the memory sizes $\mem\colon V\to\Nset$ are fixed.
A \emph{regular Defender's strategy} for $G$ is a function $\sigma\colon\dhat{V}\to\Dist(\dhat{V})$
satisfying $\sigma(v,m)(v',m')>0$ only if $v\to v'$.
Intuitively, the Defender traverses the vertices of $G$ updating memory elements and thus ``gathering''
some information about the history of visited vertices. The probability that $v'$ with information represented by $m'$ is visited after the current $v$ with information~$m$ is given by $\sigma(v,m)(v',m')$. 
A regular strategy $\sigma$ is called
\emph{deterministic-update} if $\sigma(v,m)(v',m_1)>0$ and $\sigma(v,m)(v',m_2)>0$ imply $m_1=m_2$.
This means that when the Defender is in $(v,m)$, he may randomize to choose the next vertex $v'$,
but the next memory element is then determined uniquely by $(v,m)$ and~$v'$.

For a regular strategy $\sigma$, we derive an expression $\RVal_G(\sigma)$ which corresponds to the value of $\sigma$
against an Attacker who can observe even the current memory element. This approach is consistent with the worst-case paradigm discussed in \sect{sec-intro},
as it is not clear whether the Attacker is capable of that. We show that the expression $\RVal_G(\sigma)$ is indeed a lower bound on $\Val_G(\sigma)$
(\ie, $\sigma$'s value against an Attacker who \emph{cannot} observe the current memory element, \cf. \eqref{E:val-def})
and under reasonable assumptions, they are equal (Claim~\ref{T:val}).
Since $\RVal_G(\sigma)$ is in a closed form, this makes regular strategies algorithmically workable.
Furthermore, we show that regular strategies can achieve protection \emph{arbitrarily close} to the \emph{optimal} protection achievable by unrestricted strategies. Thus, they offer a convenient trade-off between optimality and tractability.

\begin{theorem} \label{thm-reg}
Let $G$ be a patrolling graph and $\Reg$ the class of all regular Defender's strategies in $G$.
Then
\begin{equation}
	\sup_{\sigma \in \Reg} \Val_G(\sigma) = \Val_G
\end{equation}
\end{theorem}

\begin{figure}[tb]
\centering
\begin{tikzpicture}[scale=.7, every node/.style={scale=.7}, x=1cm, y=1cm]
\node [circle,fill=black,inner sep=1.5pt,label=$\lambda$] (eps) at (0,0) {}; 
\draw[thick] (eps) -- (-1.7,-4);
\draw[thick] (eps) -- (1.5,-4);
\node [circle,fill=black,inner sep=1.5pt,label=above right:$\bar{h}$,label=below left:$T_{\bar{h}}$] (barh) at (-.1,-1.5) {}; 
\draw[thick] (barh) -- +(-.3,-1) -- +(.3,-1) -- (barh);
\node [circle,fill=black,inner sep=1.5pt,label=above right:$hv$,label=below left:$T_{hv}$] (h) at (.11,-3.3) {}; 
\draw[thick] (h) -- +(-.3,-1) -- +(.3,-1) -- (h);
\draw[thick,dotted,->] (eps) -- (barh);
\draw[thick,dotted,->] (barh) -- (h);
\node [draw=none,label={optimal strategy $\gamma$}] (a) at (0,-5.2) {};
\begin{scope}[xshift=4.5cm]
\node [circle,fill=black,inner sep=1.5pt,label=$\lambda$] (eps) at (0,0) {}; 
\draw[thick] (eps) -- (-1.7,-4);
\draw[thick] (eps) -- (1.7,-4);
\node [circle,fill=black,inner sep=1.5pt,label=above right:$\bar{h}$,label=below left:$T_{\bar{h}}$] (barh) at (-.1,-1.5) {}; 
\draw[thick] (barh) -- +(-.3,-1) -- +(.3,-1) -- (barh);
\node [circle,fill=black,inner sep=1.5pt] (h) at (.11,-3.3) {}; 
\draw[thick,dotted,->] (eps) -- (barh);
\draw[thick,dotted,->] (barh) -- (h);
\draw[thick,dashed,rounded corners,->] (h) -- +(.4,0) |- (barh);
\node [draw=none,label={strategy $\sigma_\delta$}] (a) at (0,-5.2) {};
\end{scope}
\end{tikzpicture}
\caption{Folding $\gamma$ into $\sigma_\delta$.}
\label{fig-folding}
\end{figure}
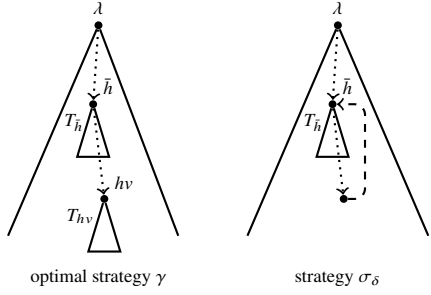

\begin{proof}[Proof (sketch)] 
First, we fix an \emph{optimal} Defender's strategy $\gamma : \calH \rightarrow \Dist(V)$  satisfying $\Val_G(\gamma) = \Val_G$ (the existence of $\gamma$ has been proven in \citep{BHKRA:Patrol-strategy-arxiv}). 

To every history $h$, we associate a finite tree $T_h$ of depth $\dm$ such that:
\begin{itemize}
\item The set of nodes contains all histories $h'$ such that $hh'$ is also a history and the length of $h'$ is bounded by $\dm$;
\item The root of $T_h$ is the empty history $\lambda$;
\item $h' \tran{x} h'v$ is an edge of $T_h$ iff the distribution $\gamma(hh')$ selects the vertex $v$ with probability~$x>0$. 
\end{itemize}
We say that histories $h,\bar{h}$ are \emph{$\delta$-similar} for a given $\delta > 0$ iff $h,\bar{h}$ end in the same vertex and the trees $T_h$ and $T_{\bar{h}}$ are the same up to $\delta$-bounded differences in edge probabilities. Observe that one can construct a fixed sequence of finite trees $T_1,\ldots,T_n$ such that \emph{every} $T_h$ is \mbox{$\delta$-similar} to some $T_i$, where the $n$ depends just on $G$ and $\delta$. Furthermore, we say that a history $h$ \emph{has $\delta$-similar prefix $\bar{h}$} iff $h = \bar{h}w$ where $h,\bar{h}$ are $\delta$-similar and the length of $w$ is larger than $\dm$. 

Let $H_\delta$ be the set of all histories $h$ such that $\Prob^\gamma(h)>0$ and no prefix of $h$ (including $h$ itself) has a $\delta$-similar prefix. Observe that the maximal length of such a history is bounded by 
$n\cdot(\dm {+}1)$ where the $n$ is defined above, and hence $H_\delta$ is a finite set. Let $\mem$ be a function assigning  $|H_\delta|$ memory elements to every vertex. To simplify our notation, we identify memory elements with the elements of $H_\delta$. Now consider the regular strategy $\sigma_\delta$ defined as follows: for every eligible pair $(v,h)$, we have that  \mbox{$\sigma_\delta(v,h)(v',h') = x$} iff $\gamma(\hat{h})(v') = x$ and $h' = \hat{h}$, where 
\begin{equation*}
	\hat{h} =
		\begin{cases}
			\bar{h} & \mbox{if $hv$ has a $\delta$-similar prefix $\bar{h}$,} \\
			hv      & \mbox{otherwise}.
		\end{cases}
\end{equation*}
Intuitively, $\sigma_\delta$ is obtained by ``folding'' the optimal
strategy~$\gamma$ after encountering a history $hv$ with a $\delta$-similar
prefix $\bar{h}$, see \fig{fig-folding}.

The proof is completed by showing that $\lim_{\delta \rightarrow 0^+}
\Val_{G}(\sigma_\delta) = \Val_G$. This is intuitively plausible, because
$\sigma_\delta$ becomes more similar to $\gamma$ for smaller $\delta$. However,
the argument also depends on the subgame-perfect property of optimal
strategies. 

Since $\sigma_\delta$ is a deterministic-update regular strategy,
Theorem~\ref{thm-reg} holds even for the subclass of deterministic-update
regular strategies.
\end{proof}

Let us fix a regular strategy $\sigma$. For convenience, we write
$\sigma(\dhat u,\dhat v)$ for $\sigma(\dhat u)(\dhat v)$.
The set of \emph{eligible edges} $\dhat E$ consists of all edges
$e\in\dhat V\times\dhat V$ such that $\sigma(e)>0$. These are
exactly the edges actually used by the Defender.
A finite sequence of eligible pairs
$u=(v_1,m_1),\dots,(v_n,m_n)$ is called an \emph{eligible path} if
$v_1,\dots,v_n$ is a path in $G$. The probability of executing $u$
is
\begin{equation}
	\Prob^\sigma(u) = \prod_{i=1}^{n-1} \sigma\bigl((v_i,m_i),(v_{i+1},m_{i+1})\bigr)
\end{equation}
For $e=((v,m),(v_1,m_1))\in\dhat E$ and $\tau\in T$, let $\Path(e,\tau)$ denote
the set of all eligible paths $(v_1,m_1),\dots,(v_k,m_k)$ such that
$v_k=\tau$, $k\ge 1$ and the traversal time of the path $v,v_1,\dots,v_k$ is
at most~$d(\tau)$.

We can now express the value of a regular strategy $\sigma$
as follows:
\begin{equation} \label{E:rval}
	\RVal_G(\sigma)
		= \alpha_{\max} -
			\max_{e\in\hat E,\,\tau\in T}\bigl\{\alpha(\tau) - \Protect^\sigma(e,\tau)\bigr\},
\end{equation}
where
\begin{equation} \label{E:P-sigma}
\Protect^\sigma(e,\tau)
= \sum_{u\in\Path(e,\tau)} \Prob^\sigma(u)\cdot \Eval(\tau\mid u).
\end{equation}

\begin{claim} \label{T:val}
Let $\sigma$ be a regular strategy.
Then
\begin{equation} \label{E:val-sigma}
\Val_G(\sigma) \geq \RVal_G(\sigma)
\end{equation}
Moreover, if the graph $(\dhat V,\dhat E)$ is strongly connected
and $\sigma$ is deterministic-update, then the above holds with equality.

\end{claim}

\begin{proof}[Proof (sketch)]
It is easy to observe that the edge lengths, the target cost $\alpha(\tau)$ and
the detection probability $\beta(\tau)$ are correctly accounted for in the
definition of $\Protect^\sigma(e,\tau)$.  The rest of the argument is the same
as in~\citep{KL:patrol-regular,DBLP:conf/atal/KlaskaKLR18}.
\end{proof}

A regular strategy $\sigma$ depends only on reasonably many variables $\sigma(e)$, $e\in\dhat E$.
Hence we identify $\sigma$ as an element of $\Rset^{|\dhat E|}$.
The function $\RVal_G\colon\Rset^{|\dhat E|}\to\Rset$ of variable $\sigma$ is differentiable up to the set where the points of maxima in \eqref{E:val-sigma} are not unique.

Claim~\ref{T:val} equips us with a closed-form formula for strategy evaluation. Hence, it enables us to apply methods from differentiable programming to obtain the value of a strategy and its sensitivity to the change of the input.

\subsection{Strategy evaluation}

We describe our algorithm that evaluates $\RVal_G$ and its gradient at
a given point $\sigma$.  According to \eqref{E:val-sigma}, we first evaluate
all the protection values $\Protect^\sigma(e,\tau)$ and their gradients
$\nabla\Protect^\sigma(e,\tau)$.  The value $\RVal_G(\sigma)$ and its gradient
should then be simply the smallest of all values and its gradient.  In
practice, we replace this minimum with its ``soft'' version. The details are
addressed in \sect{S:optim}.

\paragraph{Computing $\Protect^\sigma$ and $\nabla\Protect^\sigma$}
Given $e\in\dhat E$ and $\tau\in T$, 
a naive approach based on \emph{explicitly constructing}
$\Protect^\sigma(e,\tau)$
is inevitably inefficient, because the set $\Path(e,\tau)$ may contain
exponentially many different paths.
We overcome this problem by performing a search on the graph
during which (sub)paths with the same endpoints and the same traversal time
are ``aggregated'' into a single term, resulting in a more compact
representation of $\Protect^\sigma(e,\tau)$ and $\nabla\Protect^\sigma(e,\tau)$.
The search is guided by a min-heap $\calH$, similarly as in Dijkstra's shortest path algorithm.
However, unlike in Dijkstra's, where the search is initiated from each vertex at most once,
we must consider \emph{all} paths from $e$ to $\tau$ whose traversal time does not exceed $d(\tau)$,
and we must keep track of their probability and the corresponding gradient.
To that purpose, each item of $\calH$ corresponds to a certain \emph{set} of paths.
Moreover, instead of initiating the search at $e$, we initiate it at $\tau$
and search the graph backwards. This trick allows us to compute $\Protect^\sigma(e,\tau)$
for a given $\tau\in T$ and \emph{all} $e\in\dhat E$ at once, thereby saving
a factor of $|\dhat E|$ in the resulting time complexity.

In particular, for any $(v,m)\in\dhat V$ and $t\leq d(\tau)$, let
$\calL_{v,m,t}$ denote the set of all paths which start at $(v,m)$,
end at $(\tau,\cdot)$ and have traversal time $t$. Further, for any
$\calL\subseteq\calL_{v,m,t}$, let
\begin{equation}
\Protect^{\sigma}(\calL)
= \sum_{u \in \calL} \Prob^\sigma(u) \cdot \Eval(\tau\mid u)
\end{equation}
As a result of the search, each $\calL_{v,m,t}$ is partitioned into pairwise
disjoint sets $\calL_1,\ldots,\calL_k$ in such a way that the sets $\calL_i$
are in one-to-one correspondence with the items of $\calH$.
In \alg{alg-dynamic}, each $\calL_i$ is represented
by a tuple $(v,m,t,p,\pgrad)$ where $p$ and
$\pgrad$ correspond to the value and the gradient of
$\Protect^{\sigma}(\calL_i)$ at $\sigma$, respectively.
Then, writing $e=((v',m'),(v,m))$, the value of $\Protect^{\sigma}(e,\tau)$ (cf. \eqref{E:P-sigma})
can be computed as the sum of $\Protect^{\sigma}(\calL_{v,m,t})$ over all $t\leq d(\tau)-\tm(v',v)$
where $\Protect^{\sigma}(\calL_{v,m,t})$ is computed as the sum of $\Protect^{\sigma}(\calL_i)$
over the sets $\calL_i$ that form the partition of $\calL_{v,m,t}$.
The gradient is computed analogously.
In \alg{alg-dynamic}, we also use two auxiliary arrays $\calV$ and $\calG$ for
storing the value and the gradient of $\Protect^{\sigma}(\calL_{v,m,t})$
for all $(v,m)\in\dhat V$ and the currently examined traversal time $t$.

The body of the main loop (lines~\ref{line-initH}--\ref{line-end-z}) is
executed for every $\tau\in T$ and computes $\Protect^\sigma(e,\tau)$ and
$\nabla\Protect^\sigma(e,\tau)$ for every $e\in\dhat E$.
The correctness of the algorithm follows from the fact that
after executing line~\ref{line-set-l}, for every $(v,m)\in\dhat V$ we have that
\newcommand\pe{\;\texttt{+=}\;}
\newcommand\te{\;\texttt{*=}\;}
\newcommand\ee{\;\texttt{=}\;}
\SetAlFnt{\small}
\begin{algorithm}[h!]
	\SetAlgoLined
	\DontPrintSemicolon
	\SetKwInOut{Parameter}{parameter}\SetKwInOut{Input}{input}\SetKwInOut{Output}{output}
	\SetKwData{C}{C}\SetKwData{D}{D}\SetKwData{MX}{M}\SetKwData{f}{f}
	\SetKw{Or}{or}
	\SetKw{And}{and}
	\SetKw{Not}{not}
	\SetKwData{Array}{array of}
	\SetKwData{Rat}{Rat}
	\SetKwData{Der}{Der}
	\SetKwData{MinHeap}{min-heap of Paths}
	\SetKwData{Index}{indexed by}
	\SetKwProg{Macro}{macro}{:}{}
	\Input{A patrolling graph $G$, a regular strategy $\sigma$}
	\Output{The sets $\{\Protect^\sigma(e,\tau)\}$ and $\{\nabla\Protect^\sigma(e,\tau)\}$}
	\BlankLine
	\makebox[2em][l]{$\Protect^\sigma$} : \texttt{array}  \makebox[0em][l]{ indexed by edges $\dhat E$ and targets $T$} \;
	\makebox[2em][l]{$\nabla\Protect^\sigma$} : \texttt{array} \makebox[0em][l]{ indexed by edges $\dhat E$ and targets $T$} \;
	\makebox[1em][l]{$\calV$} : \texttt{array} \mbox{ indexed by eligible pairs $\dhat V$} \;
	\makebox[1em][l]{$\calG$} : \texttt{array} \mbox{ indexed by eligible pairs $\dhat V$} \;
	\makebox[1em][l]{$\calH$} : \texttt{min-heap} \mbox{ of tuples $(v,m,t,p,\pgrad)$ sorted by $t$}    \;
    \BlankLine
    set all elements of $\Protect^\sigma$ to $0$ and $\nabla\Protect^\sigma$ to $\vec 0$\;
	\ForEach{$\tau \in T$}{
		set $\calH$ to empty heap\; \label{line-initH}
		\ForEach{$m$ such that $(\tau,m)\in\dhat V$}{\label{initH-start}
		   $\calH.\mathit{insert}(\tau,m,0,\alpha(\tau)\beta(\tau),\vec{0})$\;
		}\label{initH-end}
	    \While{\Not $\calH.\mathit{empty}$}{\label{line-loop-H-start}
	       set all elements of $\calV$ to 0 and $\calG$ to $\vec{0}$ \;\label{line-AB-init}
	       $\ltime \ee \calH.\mathit{peek}.t$\;\label{line-set-l}
	       \Repeat{$\calH.\mathit{empty}$ \Or $\calH.\mathit{peek}.t > \ltime$}{\label{line-AB-start}
	       	   $(v,m,t,p,\der) \ee \calH.\mathit{pop}$ \;
	       	   $\calV(v,m) \pe p$ \;
	       	   $\calG(v,m) \pe \der$ \;
	       }\label{line-AB-end}
           \ForEach{$(v,m)$ such that $\calV(v,m) > 0$\label{line-vm-start}}{
               \ForEach{$e \ee ((v',m'),(v,m))\in\dhat E$}{
							 		$t \ee \tm(v',v)$\;
               	  \uIf{$\ltime + t \leq d(\tau)$}{
               	  	 $\Protect^\sigma(e,\tau) \pe \calV(v,m)$\;\label{line-A-inc}
               	  	 $\nabla\Protect^\sigma(e,\tau) \pe \calG(v,m)$\;\label{line-B-inc}
               	     $p \ee \calV(v,m) \cdot \sigma(e)$ \;\label{line-p-start}
               	     \lIf{$v'=\tau$}{
               	     	$p \te 1-\beta(\tau)$
                     }\label{line-p-end}
               	     \ForEach{$e' \in\dhat E$\label{line-der-init}}{
               	     	 $\der(e') \ee \sigma(e)\cdot\calG(v,m)(e')$\; \label{line-cder-start}
                         \lIf{$e' = e$}{
                        	$\der(e') \pe \calV(v,m)$
                        }
               	         \lIf{$v' = \tau$}{
														$\der(e') \te 1-\beta(\tau)$
                         }\label{line-cder-end}                   	
                     }\label{line-der-end}	
               	     $\calH.\mathit{insert}(v',m',\ltime+t,p,\der)$\;\label{line-H-insert}
               	  }
               }
           } \label{line-vm-end}		
        } \label{line-end-z}\label{line-loop-H-end}	
	}
	\Return $\Protect^\sigma$, $\nabla\Protect^\sigma$
	\caption{Compute $\Protect^\sigma$ and $\nabla\Protect^\sigma$}
	\label{alg-dynamic}
\end{algorithm}
\begin{equation} \label{E:p1}
	\Protect^{\sigma}(\calL_{v,m,\ell})
		= \sum_{(v,m,\ell,p,\pgrad) \in \calH} p
\end{equation}
and
\begin{equation} \label{E:p2}
	\nabla\Protect^{\sigma}(\calL_{v,m,\ell})
		= \sum_{(v,m,\ell,p,\pgrad) \in \calH} \pgrad,
\end{equation}
which can be proved by induction on the number of iterations
of the loop at lines~\ref{line-loop-H-start}--\ref{line-loop-H-end}
(the value of $\ell$ assigned at line~\ref{line-set-l}
always increases between successive iterations).

\paragraph{Complexity analysis}
Let $\Lambda$ be the total number of pairwise different $t$'s for which there exists $u \in \Path(e,\tau)$ with traversal time equal to $t$.
Note that there are at most $|\dhat E| \cdot \Lambda$ items in $\calH$, so each heap operation takes time $\mathcal{O}(\log(|\dhat E| \cdot \Lambda))$.
An analysis of the main loop (lines~\ref{line-initH}--\ref{line-loop-H-end}) reveals that the time complexity of \alg{alg-dynamic} is
\begin{equation}
	\mathcal{O}\bigl(|T|\cdot |\dhat E|^2 \cdot \Lambda \cdot (|\dhat E| + \log(|\dhat E| \cdot \Lambda)\bigr)\,.
\end{equation}
The size of $\Lambda$ plays a crucial role.
It stays reasonably small even if $G$ contains ``long'' edges.
In the worst case, $\Lambda$ can be equal to $d_{\max}$, but this is rarely the case in practice.
Clearly,  the traversal time of every $u \in \Path(e,\tau)$ is at least $t_{\min}=\min_{v\to v'} \tm(v,v')$; and for many other $t$'s between $1$ and $d_{\max}$, there may exist no $u \in \Path(e,\tau)$ with traversal time equal to $t$.
This also explains why applying the algorithm of \citet{DBLP:conf/atal/KlaskaKLR18} to the modified graph obtained from $G$ by splitting the long edges into sequences of unit-length edges is far less efficient than applying our algorithm directly to $G$. Such modification increases the number of vertices very quickly, even if $\Lambda$ is small.
This is confirmed experimentally in \sect{sec-experiments}.

\subsection{\Regstar\ algorithm}
\label{S:optim}

Having a fast, efficient and differentiable algorithm for strategy evaluation
enables us to apply gradient ascent methods.  For a given patrolling graph $G$
and fixed memory sizes $\mem$, we start with strategy $\sigma$ having its
values assigned randomly. Then, in an optimization loop,
we examine $\RVal_G(\sigma)$ and modify $\sigma$ in the direction of its gradient
until no gain is achieved (see \alg{alg-grad}).

The \Regstar algorithm constructs a Defender's strategy by running \alg{alg-grad}
repeatedly for a given number of random $\sigma$ and selecting the best outcome. This results in high-quality strategies, as experimental results confirm.

\paragraph{Normalization}

By definition, a regular strategy $\sigma$ is a bunch of probability
distributions.  A modification of $\sigma$ by an update vector $\xi$ can (and
does) violate this property. A workable solution then requires the use of a
normalization.  Our normalization procedure $N_1(\sigma)$ crops all
elements of $\sigma$ into the interval $[0,1]$ and then returns $\sigma(\dhat
v)/|\sigma(\dhat v)|$ for every $\hat v\in \dhat V$, where $|\sigma(\dhat v)| =
\sum_{(\dhat v,v)\in\dhat E} \sigma(\dhat v, v)$.  The function $\sigma\mapsto
N_1(\sigma)$ is again differentiable at almost every point $\sigma$.
Therefore, the composition of normalization and evaluation results in a
differentiable algorithm whose gradient is, by the chain rule, the composition
$\nabla\RVal_G(N_1(\sigma))\cdot\nabla N_1(\sigma)$.

Prior works \citep{KL:patrol-regular,DBLP:conf/atal/KlaskaKLR18} omitted this step assuming
that the parameter space is the set of normalized strategies.  They
implicitly modify their gradients in order that the update results in a
normalized strategy. This approach can be modelled in our setting by
considering different ``pivoted'' normalization, denoted by
$N_p(\sigma)$. It crops the values as well and, for $\dhat v\in\dhat V$,
$N_p(\sigma)(\dhat v,v)$ equals to $\sigma(\dhat v,v)$ for all $(\dhat v,v)\in\dhat
E$ except one ``pivot'', say $v_p$, for which $\sigma(\dhat v,v_p) =
1-\sum_{(\dhat v,v)\in\dhat E,v\ne v_p} \sigma(\dhat v,v)$.

Pivoting normalization $N_p$ yields very sparse gradients $\nabla N_p$ in oppose to $\nabla N_1$ and filters the signal propagation back to~$\sigma$.
This potentially slows down the optimization, as demonstrated in our experiments (see the supplementary material). 

\SetAlFnt{\small}
\begin{algorithm}[tb]
	\SetAlgoLined
	\DontPrintSemicolon
	\SetKwInOut{Parameter}{parameter}
	\SetKwInOut{Input}{input}
    \SetKwInOut{Output}{output}
	\SetKw{Or}{or}
	\SetKw{And}{and}
	\SetKwFunction{normalize}{Normalize}
	\SetKwFunction{step}{Step}
	\Input{A patrolling graph $G$, a regular strategy $\sigma$}
	\Output{A regular strategy $\sigma'$}
	\BlankLine
    \Repeat{$\RVal_G(\sigma')-p\le \text{threshold}$}{
    	$(\sigma, \nabla\sigma) = \normalize(\sigma)$ \;
			$(p, \nabla p) = \RVal_G(\sigma)$ \;
	    $\sigma' = \step(\sigma, \nabla p\cdot\nabla\sigma)$ \; \label{line-step}
    } 
		\Return $\sigma'$ \;
	\caption{Strategy optimization.}
	\label{alg-grad}
\end{algorithm}

\paragraph{Minima softening}

So far, we supposed that $\RVal_G(\sigma)$ returns a value
$\Protect^\sigma(e,\tau)$ that attains the minimal value. Therefore, only the
top candidate is taken into account in the optimization. In contrast, one can
consider more competitors $\Protect^\sigma(e,\tau)$ that are close to the
minima and optimize for them simultaneously. We implement the very same
``softening'' method as in the baseline algorithm
\citep{DBLP:conf/atal/KlaskaKLR18}.

\paragraph{Strategy update}

In each optimization step, we update the strategy $\sigma$ by
a proportion of the proper gradient $\xi=\nabla\RVal_G\cdot\nabla N_1$.
We use the same scheduling as proposed by \citet{DBLP:conf/atal/KlaskaKLR18},
where a variant of an update $\sigma'=\sigma+(1-\delta)^k\xi$ is being
used.

\section{Experiments}
\label{sec-experiments}

In many natural patrolling scenarios, the targets are distinguished geographic locations (banks, patrol stations, ATMs, tourist attractions, etc.). The connecting edges model the admissible moves of patrolling units (drones, police cars, etc.).
Such graphs naturally contain edges of varying traversal time.

All the existing strategy improvement algorithms are designed for patrolling graphs with edges of unit traversal time.
They can be applied to graphs with general topology once every ``long'' edge is replaced with a path consisting of edges of the unit length passing through fresh auxiliary vertices.
We will apply this modification to graphs when necessary, without further notice.

In the first experiment, we compare the efficiency of \Regstar against currently the best \baseline \citep{DBLP:conf/atal/KlaskaKLR18}.
In the second experiment, we demonstrate the capability of \Regstar on a real-world patrolling problem that is far beyond the
limits of \baseline. In the last experiment, we examine the impact of available memory size on the protection achieved by \Regstar.

\begin{figure}[t]
	\centering
	\includegraphics[width=\linewidth]{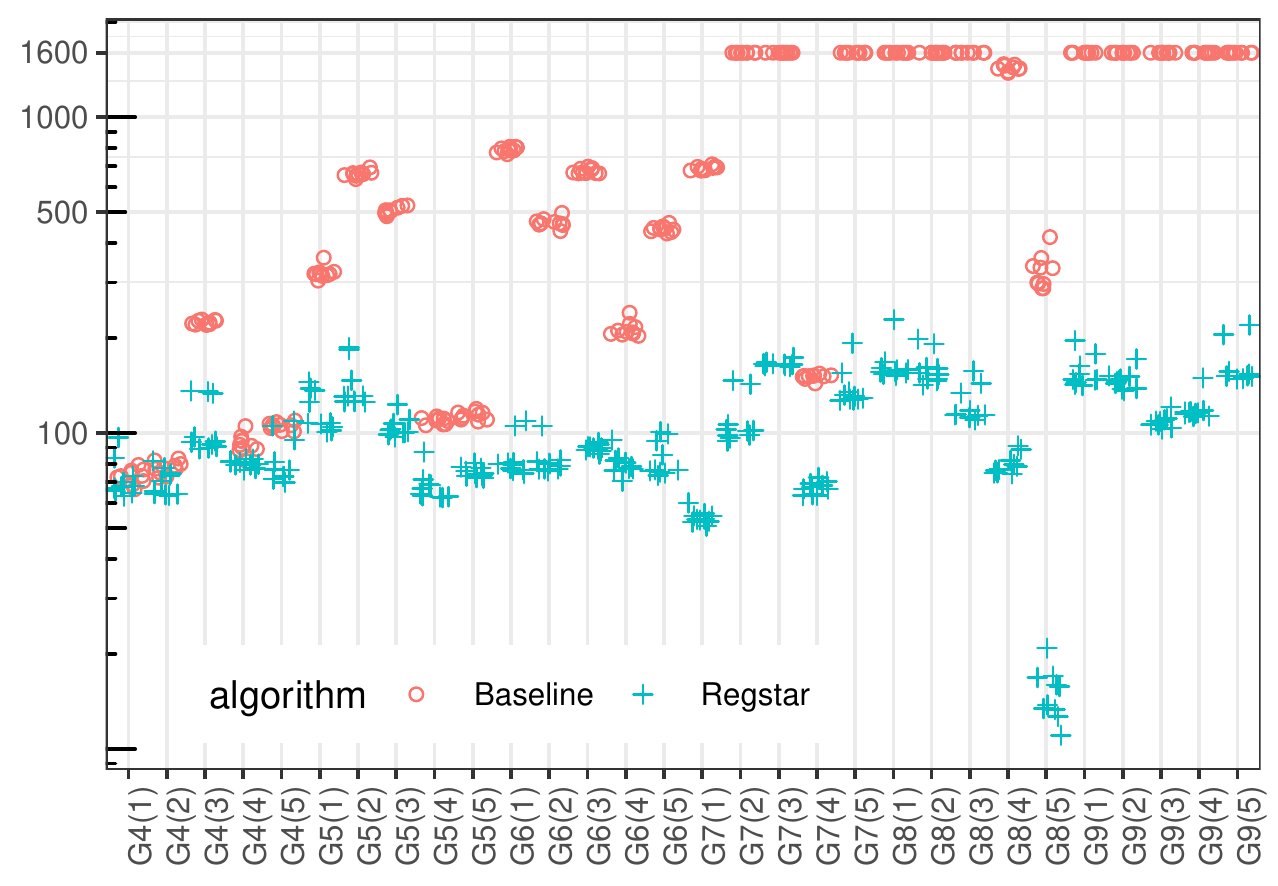}
\caption{Comparison of \Regstar against \baseline. The time (in secs., logscale) needed
to finish 50 optimization runs on various graphs is shown. The timeout is~1600~s.}
\label{fig-exp1}
\end{figure}

\subsection{Comparison to Baseline}
\label{sec-baseline}

We consider patrolling graphs where ten targets are selected randomly from $n \times n$ grid.
The traversal time between two vertices corresponds to their $L^1$ distance\footnote{We use the $L^1$ (taxicab, Manhattan) distance instead of Euclidean distance because the former is considered as a better approximation of commuting distance.}. 
For each $n = 4,\ldots,9$, we randomly select five graphs $G_n(1),\ldots,G_n(5)$.
Thus, we obtain a collection of 30 patrolling graphs.
Note that these graphs tend to contain longer edges with the increased~$n$.
The targets' cost is selected randomly between 180 and 200, and intrusion detection is perfect.
The time needed to complete an attack is the same for all targets in $G_n(i)$, and it is set to a value for which the Defender can achieve a reasonably high protection ($d(\tau)=\tm_{\max} +\tm_{\text{avg}}+3$, where $\tm_{\max}$ and $\tm_{\text{avg}}$ is the maximal and the average traversal time of an edge).
For each $G_n(i)$, we report the total running time \Regstar and \baseline need to improve the same set of 50 randomly generated initial regular strategies.
The timeout was set to 1600~s.
This is repeated 10 times for every $G_n(i)$, outcoming 20 accumulated times shown in \fig{fig-exp1}.
Note that the time scale is \emph{logarithmic}.

Observe that \Regstar terminates in about 100--200 seconds in all cases while \baseline reaches the timeout even for some $G_7(i)$ graphs and for \emph{all} $G_9(i)$ graphs.
This demonstrates that \Regstar outperforms the \baseline.
When both algorithms terminate, the achieved protection values are about the same.

\subsection{Patrolling an ATM network}
\label{sec-Montreal}

We examine a patrolling graph where the vertices correspond to selected ATMs in Montreal (\fig{F:atms}).
All parameters of the graph are chosen in the same way as in \sect{sec-baseline}, except for $d(\tau)=2*\tm_{\max} +\tm_{\text{avg}}$ and imperfect intrusion detection, which is set randomly to a value between~0.8~and~1.

Note that the corresponding graph adjustment needed to run the \baseline results in more than 
3 thousand auxiliary vertices, far beyond its limits. \citet{DBLP:conf/atal/KlaskaKLR18} claim that \baseline can solve instances with about 100 vertices. 

The results achieved by \Regstar are summarized in \tab{tab-experiment2}. 
For $m = 1,\ldots,4$ we randomly generate 100 initial regular strategies where every vertex is assigned~$m$ memory elements.
We report the best and the average protection value achieved by \Regstar, the percentage of runs for which the resulting value reached at least $90\%$ of the best value (labeled ``close''), and the average number of iterations and time needed by one \Regstar run.

Note that even for $m=4$, the algorithm can improve one strategy in about~20 mins.
The number of runs for which the optimization produces a high-value strategy is consistently very high and increases with $m$.


\begin{table}[tbp]
	\centering\small
	\caption{Analysis of \Regstar on Montreal's ATMs.}
	\begin{adjustbox}{max width=\linewidth}
	\begin{tabular}{ccccr@{\ }lr@{\ }l}
		\toprule
			$m$ & $\RVal_\text{best}$ & $\RVal_\text{avg}$ & close (\%) & \multicolumn{2}{c}{iter} & \multicolumn{2}{c}{time (s)} \\
			\midrule
			1 & 64 & $57\pm3$ & 46 & 280&$\pm$ 30  & 5&$\pm$ 1 \\
			2 & 75 & $70\pm2$ & 83 & 684&$\pm$ 41  & 79&$\pm$ 8 \\
			3 & 80 & $77\pm2$ & 100 & 1045&$\pm$ 60 & 360&$\pm$ 58 \\
			4 & 81 & $79\pm1$ & 100 & 1346&$\pm$ 75 & 1250&$\pm$ 196 \\ 
			\bottomrule
	\end{tabular}
	\end{adjustbox}
	\label{tab-experiment2}
\end{table}

\subsection{Patrolling an office building}

\begin{table*}[t]
	\centering
	\small
	\caption{Analysis of \Regstar on office buildings.}
	\begin{tabular}{c|ccr@{\ }l|ccr@{\ }l|ccr@{\ }l}
		\toprule
  \multirow{2}{*}{$m$} & 
  \multicolumn{4}{c|}{One floor}&
  \multicolumn{4}{c|}{Two floors}&
  \multicolumn{4}{c}{Three floors}\\
			    & $\RVal_\text{best}$ & close (\%) & \multicolumn{2}{c|}{time (s)} 
                & $\RVal_\text{best}$ & close (\%) & \multicolumn{2}{c|}{time (s)}
                & $\RVal_\text{best}$ & close (\%) & \multicolumn{2}{c}{time (s)}\\
			\midrule
			1
        & 27 & \phantom{0}6  & $0.02$ & $\pm$ 0.01
        & 32 & 29 & $0.6$ & $\pm$ $ 0.5$
        & 29 & 3\phantom{.0}  &  $2$ & $\pm$ $1$\\ 
			2
        & 41 & 28 & $1.2$ &$ \pm$ $ 0.5$   
        & 37 & 35 & $11$ &$ \pm$ $ 8$
        & 32 & 5\phantom{.0}  & $10$ & $ \pm$ $16$\\ 
			3
        & 44 & 84 & $4.8$ &$\pm$ $2.1$ 
        & 45 & 30 & $47$  & $\pm$ $44$
        & 44 & 0.5 & $78$ &$ \pm$ $142$\\ 
			4
        & 47 & 83 & $11.7$ &$ \pm$ $5.8$
        & 53 & \phantom{0}4  & $102$ &$ \pm$ $ 118$ 
        & 44 & 3\phantom{.0}  & $238$ &$ \pm$ $ 469$\\ 
			5
        & 47 & 75 & $24.7$ &$ \pm$ $ 14.2$
        & 56 & 12 & $236$ &$ \pm$ $ 283$ 
        & 50 & 0.5 & $286$ &$ \pm$ $ 635$\\ 
			6
        & 47 & 65 & $45.2$ &$\pm$ $ 32.8$
        & 57 & 11 & $268$ &$ \pm$ $ 420$
        & 55 & 0.5 & $641$ &$  \pm$ $ 1557$\\ 
			7
        & 51 & 45 & $74.8$ &$ \pm$ $ 64.8$ 
        & 58 & 18 & $543$ &$ \pm$ $ 877$ 
        & 53 & 1\phantom{.0}  & $1278$ &$  \pm$ $ 3414$\\ 
			8
        & 52 & 22 & $112.0$ &$ \pm$ $ 111.8$
        & 59 & \phantom{0}9  & $515$ &$ \pm$ $ 1116$
        & 56 & 1\phantom{.0}  & $1486$ &$  \pm$ $ 4466$\\ 
			\bottomrule
	\end{tabular}
	\label{tab-experiment3}
\end{table*}
We consider three office buildings with one, two, and three floors.
On each floor, there are ten offices alongside a corridor.
Stairs connect the floors on both sides of the corridors. 

\label{sec-buildings}

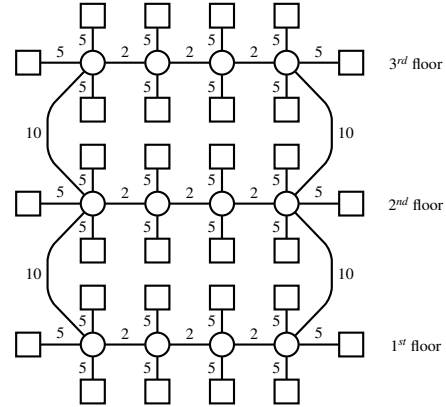
\begin{figure}[tb]
\centering
\begin{tikzpicture}[scale=.5, every node/.style={scale=0.6}, x=1.7cm, y=1.25cm]
	\foreach \x in {1,2,3,4}{
        \foreach \y in {0,2,3,5,6,8}{ 
			\node [max] (r\x\y) at (\x,\y)  {};
        };
        \foreach \d/\y/\u in {0/1/2,3/4/5,6/7/8}{ 
			\node [stoch] (c\x\y) at (\x,\y)  {};
            \draw [tran,-] (c\x\y) -- node[left] {$5$} (r\x\u);
            \draw [tran,-] (c\x\y) -- node[left] {$5$} (r\x\d);
        };
    };
    \foreach \y/\f in {1/{$1^{\textit{st}}$},4/{$2^{\textit{nd}}$},7/{$3^{\textit{rd}}$}}{
        \node [max] (r0\y) at (0,\y)  {};
        \node [max] (r5\y) at (5,\y)  {};
        \draw [tran,-] (r0\y) -- node[above] {$5$} (c1\y);
        \draw [tran,-] (r5\y) -- node[above] {$5$} (c4\y);
        \node at (6,\y) {\f \textrm{ floor}};
        \foreach \x/\r in {1/2,2/3,3/4}{
           \draw [tran,-] (c\x\y) -- node[above] {$2$} (c\r\y);  
        }
    }
    \draw [tran,-,rounded corners] (c11) -- +(-.7,1) -- node[left] {10} +(-.7,2) -- (c14);    
    \draw [tran,-,rounded corners] (c14) -- +(-.7,1) -- node[left] {10} +(-.7,2) -- (c17);    
    \draw [tran,-,rounded corners] (c41) -- +(.7,1) -- node[right] {10} +(.7,2) -- (c44);    
    \draw [tran,-,rounded corners] (c44) -- +(.7,1) -- node[right] {10} +(.7,2) -- (c47);   
\end{tikzpicture}
\caption{A building with three floors connected by stairs.}
\label{fig-building}
\end{figure}

The building with three floors is shown in \fig{fig-building}.
The squares represent the offices, and the circles represent the corridor locations where the Defender may decide to visit the neighbouring offices.
The ``long'' edges represent stairs.
Every office's cost is set to 100, and the probability of successful intrusion detection is 0.9.
The time needed to complete an intrusion is set to 100, 200, and 300 for the building with one, two, and three floors, respectively. 

The outcomes achieved by \Regstar are summarized in \tab{tab-experiment3}.
For every building and every \mbox{$m\in\{1,\ldots,8\}$}, 200 runs were processed with
all nodes having $m$ memory elements.
We report the best value found, the percentage of runs for which the resulting value reached at least $90\%$ of the best value (labeled ``close''), and the average runtime.

We observe that the achieved protection substantially increases with more memory elements.
This is because the considered patrolling graphs are relatively sparse, and ``remembering'' the history of previously visited vertices is inevitable for arranging optimal moves.
In contrast, the patrolling graphs of ATM networks presented in \sect{sec-Montreal} are fully connected, and the extra memory elements have not brought many advantages.

The \Regstar algorithm discovers the relevant information about the history fully automatically.
To examine this capability, we again consider the building with one floor, but we change the probability of successful intrusion detection to~1 and increase the attack time to~112 units.
This is \emph{precisely} the time the Defender needs to visit every target before an arbitrary attack is completed to achieve a \emph{perfect} protection equal to 100.
The Defender must schedule his walk to visit every office precisely once and with appropriate timing.

This experiment (see \tab{tab-experiment4}) show that \Regstar can indeed discover this ``clever'' walk.
For every $m\in\{1,\ldots,8\}$, 500 runs were processed.
The third column shows the percentage of runs resulting in a perfect protection strategy.
Observe that \emph{four} memory elements are sufficient to achieve perfect protection, and the chance of discovering a perfect strategy further increases with more memory elements.

\begin{table}[t]
	\centering
	\footnotesize
	\caption{Reaching perfect protection for a one-floor office building with tight attack time.}
	\begin{tabular}{cccr@{\ }l}
		\toprule
		 $m$ & $\RVal_\text{best}$ & close (\%) & \multicolumn{2}{c}{time (s)} \\
			\midrule
			1 & \phantom{0}34  & \phantom{11.}0    & 0.05 &$\pm$ 0.03   \\ 
			2 & \phantom{0}50  & \phantom{11.}0    & 2.1  &$\pm$ 1.0   \\ 
			3 & \phantom{0}55  & \phantom{11.}0    & 8.6  &$\pm$ 4.8   \\ 
			4 & 100 & \phantom{1}1.2  & 20   &$\pm$ 13   \\ 
			5 & 100 & \phantom{1}5.0  & 37   &$\pm$ 31   \\ 
			6 & 100 & 12.2 & 73   &$\pm$ 67   \\ 
			7 & 100 & 11.8 & 107  &$\pm$ 122   \\ 
			8 & 100 & \phantom{1}8.6  & 119  &$\pm$ 180   \\ 
			\bottomrule
	\end{tabular}
	\label{tab-experiment4}
\end{table}

\section{Conclusion}
\label{sec-concl}

Our efficient and differentiable regular-strategy evaluation algorithm proved to apply to patrolling graphs with arbitrary edge lengths and imperfect intrusion detection.
The experimental results are encouraging and indicate that high-quality Defender's strategies can be constructed by optimization methods in a reasonable time for real-world 
scenarios.

Our experiments also show that the current optimization methods do not often converge to the best possible strategy.
This suggests that an extra performance could be gained in exploration and improvement on the optimization side.

In Theorem~\ref{thm-reg}, we proved that regular strategies \emph{approximate} the optimal strategy to arbitrary precision $\varepsilon>0$. The needed memory size depends on $\varepsilon$. We evaluated various memory sizes $m$ experimentally from $m=1$ up to the highest computable values in \sect{sec-Montreal} and \sect{sec-buildings}. Finding the best $m$ (or even proving that regular strategies achieve the optimality) is a challenging open problem. 

\paragraph*{Acknowledgment}

Research was sponsored by the Army Research Office and was accomplished under
Grant Number W911NF-21-1-0189.
This work was supported from Operational Programme Research, Development and
Education - Project Postdoc2MUNI (No.\ CZ.02.2.69/0.0/0.0/18\_053/0016952).

\paragraph*{Disclaimer}

The views and conclusions contained in this document are those of the authors
and should not be interpreted as representing the official policies, either
expressed or implied, of the Army Research Office or the U.S.\ Government. The
U.S.\ Government is authorized to reproduce and distribute reprints for
Government purposes notwithstanding any copyright notation herein.

\clearpage

\appendix

\section{Proof of Theorem~1}
\label{app-regular}

In this section, we show that
Theorem~1 holds even for a special type of
\emph{deterministic-update} regular strategies where, for all $(v,m) \in
\dhat{V}$ and $v' \in V$, there is at most one $m'$ such that
$\sigma((v,m),(v',m')) > 0$.

For the rest of this section, we fix a patrolling graph $G =
(V,T,E,\tm,d,\alpha,\beta)$. We say that a strategy $\gamma$ is \emph{optimal}
if $\Val_G(\gamma) = \Val_G$. The existence of optimal strategies in patrolling
games has been proven in \cite{BHKRA:Patrol-strategy-arxiv}.

For a strategy $\gamma$ and $o = v_1,\ldots, v_n$, $v_n {\rightarrow} v_{n+1}
\in \Omega$, let $\gamma[o]$ be a strategy that starts in $v_n$ by selecting
the edge $v_n {\rightarrow} v_{n+1}$ with probability one, and for every finite
path of the form $v_n,v_{n+1},\ldots,v_{n+k+1}$ where $k\ge 1$ we have that
$\gamma[o](v_n,v_{n+1},\ldots,v_{n+k+1}) = \gamma(v_1,\ldots,v_{n+k+1})$. 

First, we need the following lemma.

\begin{lemma} \label{lem-aux}
Let $\gamma$ be an optimal strategy and $o \in \Omega$ such that
$\Prob^\gamma(o) > 0$. Then $\gamma[o]$ is optimal. 
\end{lemma}

\begin{proof}
For every $\ell \ge 1$, let $\Omega(\ell)$ be the set of all observations $o =
v_1,\ldots,v_\ell,v_\ell{\rightarrow}v_{\ell+1}$  such that $\Prob^{\gamma}(o)
> 0$. Clearly, for every fixed $\ell \ge 1$ we have that
\begin{align*}
	\sum_{o \in \Omega(\ell)} \Prob^{\gamma}(o) \cdot \Val_G(\gamma[o])
		\le \Val_{G}
\end{align*}
because $\Val_G(\gamma[o]) \le \Val_G$ and $\sum_{o \in \Omega(\ell)}
\Prob^{\gamma}(o) = 1$. We show that 
\begin{equation} \label{neq-val}
	\Val_{G}
		\le \sum_{o \in \Omega(\ell)} \Prob^{\gamma}(o) \cdot \Val_G(\gamma[o]) 
\end{equation} 
which implies $\Val_{G} = \sum_{o \in \Omega(\ell)} \Prob^{\gamma}(o) \cdot
\Val_G(\gamma[o])$, and hence $\Val_G(\gamma[o]) = \Val_{G}$ for every $o \in
\Omega(\ell)$. 

It remains to prove~\eqref{neq-val}. Since $\Val_{G}  = \Val_{G}(\gamma)$, it
suffices to show that, for an arbitrarily small $\varepsilon > 0$, 
\begin{align*}
	\Val_{G}(\gamma)
		\le \varepsilon + \sum_{o \in \Omega(\ell)}
			\Prob^{\gamma}(o) \cdot \Val_G(\gamma[o]) \,. 
\end{align*}
For every $o \in \Omega(\ell)$, let $\pi_o$ be an Attacker's strategy such that
$\EU_D(\gamma[o],\pi_{o}) \le \Val_G(\gamma[o])+ \varepsilon$. Consider
another Attacker's strategy $\hat{\pi}$ waiting for the first $\ell$ moves and
then ``switching'' to an appropriate $\pi_o$ according to the corresponding
observation. Then,
\begin{align*}
  \Val_{G}(\gamma)
		& \le \EU_D(\gamma,\hat{\pi})
			\\
  	& \le \sum_{o \in \Omega(\ell)} \Prob^{\gamma}(o) \cdot \EU_D(\gamma[o],\pi_{o})
			\\
  	& \le \sum_{o \in \Omega(\ell)} \Prob^{\gamma}(o) \cdot (\Val_{G}(\gamma[o]) + \varepsilon)
			\\
  	& = \varepsilon + \sum_{o \in \Omega(\ell)} \Prob^{\gamma}(o) \cdot \Val_{G}(\gamma[o]) \,.
\end{align*}
This completes the proof of Lemma~\ref{lem-aux}.
\end{proof}

\begin{proof}[Proof of Theorem~1]
We show that for an arbitrarily small $\varepsilon > 0$, there exists a
deterministic-update regular strategy $\sigma_\varepsilon$ such that
$\Val_{G}(\sigma_\varepsilon) \ge \Val_G - \varepsilon$.
	
Let $\dm = \max_{t\in T} d(t)$, and $\am = \max_{t\in T} \alpha(r)$ and let an
optimal Defender's strategy $\gamma$ be fixed.
	
We say that two non-empty finite paths $h,h' \in \hist$ are
\emph{$\delta$-similar}, where $\delta > 0$, if the following conditions are
satisfied:
\begin{itemize}
	\item $h$ and $h'$ end with the same vertex $v$,
	\item $\Prob^{\gamma}(h) > 0$, $\Prob^{\gamma}(h') > 0$,
	\item for every target $t$ and every $\ell \in \{1,\ldots,\dm\}$, the
	probabilities that $\gamma$ successfully detects an ongoing attack at $t$ in
	at most $\ell$ time units after executing the histories $h$ and $h'$ differ
	at most by $\delta$. 
\end{itemize}
Note that there are only \emph{finitely many} pairwise non-$\delta$-similar
histories. More precisely, their total number is bounded from above by
$|V|\cdot\left(\lceil \delta^{-1} \rceil\right)^{\dm \cdot |T|}$.
    
Let us fix an arbitrarily small $\varepsilon > 0$, and let $\delta =
\varepsilon/\am$. Furthermore, let $\kappa = |V|\cdot\left(\lceil \delta^{-1}
\rceil\right)^{\dm \cdot |T|}$. We construct a regular deterministic-update
strategy $\sigma_\varepsilon$ as follows:
\begin{itemize}
	\item Let $H_\delta$ be the set of all finite paths $h$ of length at most
	$\kappa \cdot \dm$ such that $\Prob^{\gamma}(h)>0$ and for all proper
	prefixes $h',h''$ of $h$ whose length is a multiple of $\dm$ we have that if
	$h',h''$ are $\delta$-similar, then $h' = h''$.  For notation simplification,
	from now on we identify memory elements with such finite paths.
	\item For every eligible pair $(v,h)$, the distribution
	$\sigma_{\varepsilon}(v,h)$ is determined in the following way:
	\begin{itemize}
		\item If the length of $h$ is a multiple of $\dm$ and there is a proper
		prefix $h'$ of $h$ where the length of $h'$ is also a multiple of $\dm$ and
		the histories $h,h'$ are $\delta$-similar, then $\sigma_{\varepsilon}(v,h)
		= \sigma_{\varepsilon}(v,h')$ (since $h'$ is shorter than $h$, we may
		assume that $\sigma_{\varepsilon}(v,h')$ has already been defined).
		\item Otherwise, $\sigma_{\varepsilon}(v,h)$ is a distribution $\mu \in
		\Dist(V {\times} H_\delta)$ such that $\mu(v',hv') =
		\Prob^{\gamma}(hv')/\Prob^{\gamma}(h)$ for every vertex $v'$ such that $hv'
		\in H_\delta$. For the other pairs of $V \times H_\delta$, the distribution
		$\mu$ returns zero.  
	\end{itemize}
	\item The initial distribution assigns $\gamma(\lambda)(v)$ to every $(v,v)
	\in V \times H_\delta$. For the other pairs of $V \times H_\delta$, the
	initial distribution returns zero.
\end{itemize}
Intuitively, the strategy $\sigma_\varepsilon$ mimics the optimal strategy
$\gamma$, but at appropriate moments ``cuts'' the length of the history stored
in its memory and starts to behave like $\gamma$ for this shorter history.
These intermediate ``switches'' may lower the overall protection, but since the
shorter history is $\delta$-similar to the original one, the impact of these
``switches'' is very small. 
    
More precisely, we show that, for an arbitrary Attacker's strategy $\pi$,
$\EU_D(\sigma_{\varepsilon},\pi) \ge \EU_D(\gamma,\pi) - \varepsilon$. Since
$\gamma$ is optimal, we obtain $\EU_D(\sigma_{\varepsilon},\pi) \ge \Val_{G} -
\varepsilon$, hence $\Val_{G}(\sigma_\varepsilon) \ge \Val_G - \varepsilon$ as
required. For the rest of this proof, we fix an Attacker's strategy $\pi$.  For
every target $\tau$, let $\pi_\tau$ be an Attacker's strategy such that
$\pi_\tau(u{\rightarrow}v) = \attack_\tau$ for every edge $u{\rightarrow}v$,
\ie, $\pi_\tau$ attacks~$\tau$ immediately.  Furthermore, let
$\Attack(\pi,\tau)$ be the set of all observations $o$ such that
$\Prob^{\sigma_\varepsilon}(o) > 0$ and $\pi(o) = \attack_\tau$. We have the
following: 
\begin{align*}
	& \EU_A(\sigma_{\varepsilon},\pi)
			\\
		&\quad = \sum_{\tau \in T} \sum_{o \in \Attack(\pi,\tau)} \Prob^{\sigma_\varepsilon}(o)
					\cdot (\alpha(\tau)-\Protect^{\sigma_\varepsilon}(\tau\mid o))
				\\
		&\quad = \sum_{\tau \in T} \sum_{o \in \Attack(\pi,\tau)}\Prob^{\sigma_\varepsilon}(o)
					\cdot \EU_A(\sigma_\varepsilon[o],\pi_\tau)
\end{align*}
Here, $\sigma_\varepsilon[o]$, where $o = v_1,\ldots, v_n$, $v_n {\rightarrow}
v_{n+1}$, is a strategy that starts in $v_n$ by executing the edge $v_n
{\rightarrow} v_{n+1}$, and then behaves identically as $\sigma_\varepsilon$
after the history $o$ (since $\sigma_\varepsilon$ is deterministic-update, the
associated memory elements are determined uniquely by $o$). 

Now, realize that for every $o \in \Attack(\pi,\tau)$, there exists an
observation $o'$ (stored in the finite memory of~$\sigma_\varepsilon$) such
that $\Prob^{\gamma}(o')>0$ and the strategy $\sigma_\varepsilon[o]$
``mimics'' the strategy $\gamma[o']$ until the finite path stored in the memory
of $\sigma_\varepsilon$ is ``cut'' into a shorter path in the way described
above. Since at most one such ``cut'' is performed during the first $\dm$ steps
and the shorter path obtained by the cut is $\delta$-similar to the original
one, we obtain that the difference between
$\EU_A(\sigma_\varepsilon[o],\pi_\tau)$ and $\EU_A(\gamma[o'],\pi_\tau)$ is at
most~$\varepsilon$.

By Lemma~\ref{lem-aux}, we obtain $\EU_D(\gamma[o'],\pi_\tau) \ge \Val_G$,
hence $\EU_A(\gamma[o'],\pi_\tau) \le \am-\Val_G$ and
$\EU_A(\sigma_\varepsilon[o],\pi_\tau) \le \am-\Val_G+\varepsilon$.  This
gives
\begin{align*}
	& \EU_A(\sigma_{\varepsilon},\pi)
			\\
		& \quad\le \sum_{\tau \in T} \sum_{o \in \Attack(\pi,\tau)}\Prob^{\sigma_\varepsilon}(o)
			\cdot (\am-\Val_G+\varepsilon)
				\\
		& \quad= (\am-\Val_G+\varepsilon)\cdot \sum_{\tau \in T}
				\sum_{o \in \Attack(\pi,\tau)}\Prob^{\sigma_\varepsilon}(o)
				\\
		& \quad\le \am - \Val_G + \varepsilon
\end{align*}
since the sum is equal to the probability that $\pi$ attacks at all against
$\sigma_\varepsilon$, which is at most $1$.  Hence,
$\EU_D(\sigma_{\varepsilon},\pi) \ge \Val_G - \varepsilon$ and we are done. 
\end{proof}

\end{document}

%% file: macros.tex
\usepackage[american]{babel}

\usepackage{amsthm}
\usepackage{mathtools}
\usepackage{flushend}

\newtheorem{theorem}{Theorem}
\newtheorem{claim}{Claim}
\newtheorem{lemma}[theorem]{Lemma}

\newcommand{\tm}{\mathit{time}}
\newcommand{\Nset}{\mathbb{N}}

\newcommand{\Rset}{\mathbb{R}}

\newcommand{\calL}{\mathcal{L}}

\newcommand{\calH}{\mathcal{H}}

\newcommand{\calV}{\mathcal{V}}
\newcommand{\calG}{\mathcal{G}}
\newcommand{\wait}{\mathit{wait}}
\newcommand{\attack}{\mathit{attack}}

\newcommand{\EU}{\mathbb{E}U}

\newcommand{\Val}{\operatorname{Val}}
\newcommand{\RVal}{\operatorname{RVal}}

\newcommand{\Obs}{\Omega}

\newcommand{\Prob}{\operatorname{Prob}}

\newcommand{\Dist}{\mathit{Dist}}
\newcommand{\Attack}{\mathit{Att}}

\newcommand{\hist}{\mathcal{H}}

\newcommand{\Path}{\textit{Path}}

\newcommand{\Protect}{\mathbf{P}}
\newcommand{\ltime}{\ell}

\newcommand{\dhat}[1]{\smash{\hat{#1}}}
\newcommand{\dm}{d_{\max}}
\newcommand{\am}{\alpha_{\max}}

\newcommand{\alg}[1]{Alg.~\ref{#1}}
\newcommand{\fig}[1]{Fig.~\ref{#1}}
\newcommand{\tab}[1]{Tab.~\ref{#1}}
\newcommand{\sect}[1]{Sec.~\ref{#1}}

\newcommand{\Reg}{\mathit{Reg}}

\newcommand\pgrad{p_\text{grad}}
\let\der\pgrad
\newcommand\mem{\operatorname{mem}}
\let\hat\widehat

\hypersetup{
	colorlinks=true,
	linkcolor=teal,
	citecolor=violet
}

\usepackage{xspace}
\makeatletter
\DeclareRobustCommand\onedot{\futurelet\@let@token\@onedot}
\def\@onedot{\ifx\@let@token.\else.\null\fi\xspace}

\newcommand\eg{{e.g}\onedot} 
\newcommand\ie{{i.e}\onedot} 
\newcommand\cf{{c.f}\onedot}

\makeatother

\newcommand{\Regstar}{\textsc{Regstar}\xspace}
\newcommand{\baseline}{\textsc{Baseline}\xspace}

\newcommand{\tran}[1]{\stackrel{\raisebox{-.3ex}{\scriptsize$#1$}}{\rightarrow}}